\documentclass[10pt, twocolumn]{article}

\usepackage{lmodern}

\usepackage{amsmath, amssymb, amsthm}
\usepackage{booktabs}
\usepackage{algorithm}
\usepackage{algpseudocode}
\usepackage{graphicx}
\usepackage{geometry}
\usepackage[hidelinks]{hyperref}
\usepackage{caption}
\usepackage{cite}
\usepackage{array}
\usepackage{adjustbox}
\usepackage{listings}
\usepackage{xcolor}
\usepackage{enumitem}
\usepackage{microtype}
\usepackage{url}
\usepackage{flafter}
\usepackage[section,above,below]{placeins}


\graphicspath{{figs/}}

\newtheorem{theorem}{Theorem}
\newtheorem{lemma}{Lemma}
\newtheorem{proposition}{Proposition}

\newcommand{\algo}[1]{\textsc{#1}}
\DeclareRobustCommand{\nbhy}{\nobreakdash-\hspace{0pt}}
\newcommand{\LRH}{\algo{LRH}}
\newcommand{\HRW}{\algo{HRW}}
\newcommand{\MPCH}{\algo{MPCH}}
\newcommand{\CH}{\algo{CH}}
\newcommand{\BLCH}{\algo{CH\nbhy BL}}
\newcommand{\AnchorHash}{\algo{AnchorHash}}

\setlist{nosep, leftmargin=1.5em}

\title{\textbf{Local Rendezvous Hashing:\\[0.3ex]
Bounded Loads and Minimal Churn\\
via Cache\nbhy Local Candidates}}

\author{
    Yongjie Guan \\
    Zhejiang University of Technology \\[0.6ex]
}
\date{\today}

\begin{document}
\maketitle

\begin{abstract}
Consistent hashing is fundamental to distributed systems. Classical ring\nbhy based schemes offer stability under updates but can exhibit a peak\nbhy to\nbhy average load ratio (PALR) of $\Theta(\ln n)$ in the single\nbhy token regime; achieving PALR of $1{+}\varepsilon$ requires $\Theta(\ln n/\varepsilon^2)$ virtual nodes per physical node in widely\nbhy used analyses \cite{appleton2015mpch}. Multi\nbhy probe consistent hashing (\MPCH) achieves near\nbhy perfect balance with $O(1/\varepsilon)$ probes \cite{appleton2015mpch}, but its probes typically translate into scattered memory accesses.

We introduce \textbf{Local Rendezvous Hashing (\LRH)}, which restricts Highest Random Weight (\HRW) selection \cite{thaler1998hrw,thaler1996rendezvous} to a contiguous window of $C$ \emph{distinct} physical neighbors on a ring. Precomputed \emph{next\nbhy distinct} offsets make candidate enumeration \emph{exactly $C$ ring steps} (distinct nodes), while the overall lookup remains $O(\log|\mathcal{R}|{+}C)$ due to the initial binary search. In a large\nbhy scale benchmark with $N{=}5000$ nodes, $V{=}256$ vnodes/node, $K{=}50$M keys, $C{=}8$ (20 threads), \LRH{} reduces PALR (Max/Avg) from $1.2785$ to $1.0947$ and empirically enforces \emph{ScanMax$=C$} under fixed\nbhy candidate enumeration. Under fixed\nbhy candidate liveness failover, \LRH{} achieves 0\% excess churn. Compared to \MPCH (8 probes), \LRH{} achieves ${\approx}6.8\times$ higher throughput (60.05 vs 8.80 M keys/s) while approaching \MPCH{}'s load balance. A microbenchmark further shows that speeding up \MPCH{} probe generation by $4.41\times$ improves assign-only throughput by only $1.06\times$, confirming that \MPCH{} remains dominated by $P \times$ lower\nbhy bound ring memory traffic rather than hash arithmetic.
\end{abstract}

\section{Introduction}

Partitioning keys uniformly across a changing set of nodes underpins sharded caches, distributed storage, and load balancers (e.g., Dynamo \cite{decandia2007dynamo}, Maglev \cite{eisenbud2016maglev}). An ideal assignment function must jointly satisfy:
\begin{enumerate}
    \item \textbf{Uniformity:} load should be balanced (PALR, P99/Avg).
    \item \textbf{Minimal churn:} only necessary keys move on changes.
    \item \textbf{Performance:} lookups must be fast and cache\nbhy friendly.
\end{enumerate}

In practice these objectives conflict. Ring consistent hashing \cite{karger1997consistent} is stable, but can be imbalanced without enough vnodes \cite{appleton2015mpch}. \MPCH{} improves balance using multiple probes \cite{appleton2015mpch}, but implies additional probe lookups and memory touches. Maglev \cite{eisenbud2016maglev} uses a lookup table for near\nbhy perfect balance yet explicitly tolerates a small number of table disruptions under backend changes. Recent consistent hashing systems such as \AnchorHash{} aim for high performance and strong consistency properties under arbitrary membership changes \cite{mendelson2021anchorhash}.

\paragraph{Our goal.}
We focus on a common setting in high\nbhy throughput data planes: \emph{fast per\nbhy key lookup with bounded work, cache locality, and minimal churn under liveness failures}. We propose \LRH: keep the ring topology but \emph{smooth local imbalance} by running a small \HRW{} election among $C$ cache\nbhy local distinct candidates.

\paragraph{Contributions.}
\begin{itemize}
    \item \textbf{\LRH{} algorithm:} \HRW{} election within a ring\nbhy local window of $C$ distinct nodes.
    \item \textbf{Next\nbhy distinct offsets:} enforce bounded distinct candidate enumeration in exactly $C$ ring steps.
    \item \textbf{Liveness failures with minimal churn:} fixed\nbhy candidate filtering yields 0\% excess churn under topology\nbhy fixed failures.
    \item \textbf{Weighted extension:} topology--weight decoupling via weighted \HRW{} as standardized in an IETF draft \cite{ietfWeightedHRW}.
    \item \textbf{Large\nbhy scale evaluation:} $N{=}5000$, $V{=}256$, $K{=}50$M, multiple failure sizes and repeats, with detailed churn and concentration metrics.
\end{itemize}

\section{Background and Problem Setting}

\subsection{Two kinds of ``changes'': membership vs liveness}

We distinguish two operational regimes:

\paragraph{Membership change.}
Nodes are added/removed from the configuration. Data structures may change (ring rebuild, table rebuild, etc.). This may induce churn beyond the theoretical minimum, depending on algorithm design \cite{eisenbud2016maglev,mendelson2021anchorhash}.

\paragraph{Liveness change.}
Nodes temporarily fail/recover, while the configured topology remains fixed (e.g., alive mask changes). Many systems strongly prefer minimal movement under such failures because failures may be frequent and transient.

Our strongest churn guarantee is for \emph{liveness changes} with fixed topology (Theorem~\ref{thm:churn}).

\subsection{Ring hashing, vnodes, and imbalance}

Ring consistent hashing maps nodes (or vnodes) to tokens on the unit circle and assigns keys to the next token clockwise \cite{karger1997consistent}. Analyses in \MPCH{} show that in the one\nbhy token setting the maximum gap is $\Theta(\ln n / n)$, yielding PALR $\Theta(\ln n)$; and to achieve PALR $\le 1{+}\varepsilon$ requires $\Theta(\ln n/\varepsilon^2)$ vnodes per node \cite{appleton2015mpch}. Large vnode state can stress memory and caches in high\nbhy throughput implementations.
Our vnode sweep (Section~\ref{sec:ring-vn-sweep}) empirically quantifies this trade-off: higher $V$ improves balance but substantially increases ring build cost and reduces lookup throughput.

\subsection{HRW (Rendezvous hashing)}

\HRW{} \cite{thaler1998hrw} assigns each key to the node maximizing a hash\nbhy derived score. It achieves strong balance but costs $O(N)$ work per key if performed over all nodes.

\subsection{MPCH and Maglev}

\MPCH{} \cite{appleton2015mpch} reduces imbalance using $K$ probes per key. In a typical implementation, each probe implies additional index lookups; performance can suffer when the probes translate into scattered memory accesses (a practical point also highlighted when contrasting ring/vnode tables vs other approaches) \cite{lamping2014jump}.

Maglev \cite{eisenbud2016maglev} builds a lookup table of size $M$ for fast dispatch, achieving near\nbhy perfect balance; it tolerates a small disruption rate under backend changes and discusses the trade\nbhy off between table size, disruption, and build cost.

\subsection{Where \LRH{} fits among modern alternatives}

Two influential directions:

\paragraph{\BLCH{} (bounded loads).}
\BLCH{} augments \CH{} with capacity constraints and forwarding to bound maximum load while controlling expected moves under updates \cite{mirrokni2018bounded}.

\paragraph{\AnchorHash{}.}
\AnchorHash{} achieves strong consistency and high performance under arbitrary membership changes without large vnode state \cite{mendelson2021anchorhash}. 
\LRH{}, in contrast, is \emph{ring-preserving}: it keeps an existing token ring and adds only a per-entry next-distinct offset plus a $C$-way \HRW{} election over ring-local distinct candidates. 
This targets deployments where failures are treated as liveness changes (fixed topology) and ring compatibility/cache-local bounded work are primary constraints.

\LRH{} is complementary: it keeps the ring topology (ease of integration) but introduces \emph{cache\nbhy local elections} to improve balance with fixed bounded work.

\section{Local Rendezvous Hashing (\LRH)}

\subsection{Data Structure}

Let $\mathcal{N}=\{0,\dots,N{-}1\}$ be physical node ids. The ring $\mathcal{R}$ is a sorted array of $|\mathcal{R}|=N \cdot V$ entries:
\[
\mathcal{R}[i] = \bigl(\mathit{token}_i,\; \mathit{node}_i,\; \delta_i\bigr),
\]
sorted by $\mathit{token}_i$. The \textbf{next\nbhy distinct offset} $\delta_i$ satisfies:
\[
\mathcal{R}\bigl[(i{+}\delta_i)\bmod |\mathcal{R}|\bigr].\mathit{node} \ne \mathcal{R}[i].\mathit{node},
\]
and is the smallest such positive offset (wrapping around).

\subsection{Lookup}

Given key $k$ and candidate count $C$:

\begin{enumerate}
  \item Compute $h=\algo{HashPos}(k)$ and locate $idx=\algo{LowerBound}(\mathcal{R},h)$.
  \item Enumerate $C$ distinct candidates by walking $\delta$ offsets.
  \item Select the candidate with the maximum \HRW{} score (or weighted score).
\end{enumerate}

\begin{algorithm}[t]
\caption{\LRH{} Lookup}
\label{alg:lookup}
\begin{algorithmic}[1]
\Function{Lookup}{$k$, $C$}
    \State $h \gets \textsc{HashPos}(k)$
    \State $idx \gets \textsc{LowerBound}(\mathcal{R}, h)$
    \State $best \gets \texttt{null}$;\; $best\_s \gets -\infty$
    \For{$t \gets 1$ \textbf{to} $C$}
        \State $e \gets \mathcal{R}[idx]$
        \State $s \gets \textsc{HashScore}(k, e.\mathit{node})$
        \If{$s > best\_s$}
            \State $best\_s \gets s$;\; $best \gets e.\mathit{node}$
        \EndIf
        \State $idx \gets (idx + e.\delta) \bmod |\mathcal{R}|$
    \EndFor
    \State \Return $best$
\EndFunction
\end{algorithmic}
\end{algorithm}

\paragraph{Bounded distinct enumeration.}
The next\nbhy distinct offsets guarantee \emph{exactly $C$ ring steps} to enumerate $C$ distinct physical nodes, independent of vnode clustering. The remaining cost is the initial lower\nbhy bound search; hence total complexity is $O(\log|\mathcal{R}|{+}C)$.

Figure~\ref{fig:scanmax-c} shows that, under fixed-candidate enumeration, \LRH{} empirically enforces ScanMax$=C$ across failure sizes.

\begin{figure}[tb]
\centering
\includegraphics[width=\columnwidth]{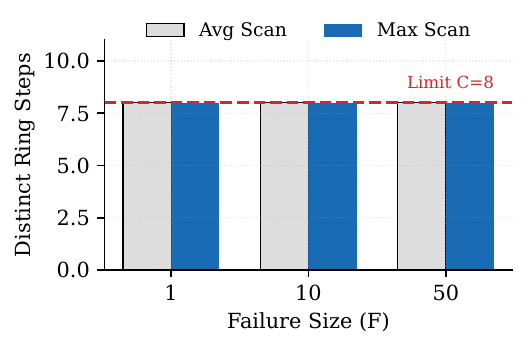}
\caption{Bounded distinct enumeration evidence: under fixed-candidate enumeration, \LRH{} empirically enforces ScanMax$=C$ (here $C{=}8$) across failure sizes.}
\label{fig:scanmax-c}
\end{figure}

\subsection{Building next\texorpdfstring{\nbhy}{-}distinct offsets}

Offsets are computed once per ring build (or rebuild). Let $R[i].node$ denote the node id.

\begin{algorithm}[t]
\caption{BuildNextDistinctOffsets}
\label{alg:builddelta}
\begin{algorithmic}[1]
\Function{BuildNextDistinctOffsets}{$R$}
    \State $m \gets |R|$
    \State $j \gets 1$
    \For{$i \gets 0$ \textbf{to} $m{-}1$}
        \If{$j \le i$} \State $j \gets i{+}1$ \EndIf
        \While{$R[j \bmod m].node = R[i].node$}
            \State $j \gets j{+}1$
        \EndWhile
        \State $R[i].\delta \gets (j{-}i)$
    \EndFor
\EndFunction
\end{algorithmic}
\end{algorithm}

This is $O(|R|)$ time and $O(1)$ extra memory.

\subsection{Weighted nodes (topology--weight decoupling)}

To support heterogeneous capacities without rebuilding the ring, we use weighted \HRW{} as standardized in an IETF draft \cite{ietfWeightedHRW} and commonly derived via exponential-race style constructions \cite{schindelhauer2005wdht,efraimidis2006wrs}. One equivalent form is:
\[
\arg\min_n \frac{-\ln(u_{k,n})}{w_n}, \quad u_{k,n} \sim \mathrm{Unif}(0,1].
\]
Weights $w_n$ reside in a separate array; updating weights is $O(1)$ and does not modify tokens.

\subsection{Failure handling: fixed\texorpdfstring{\nbhy}{-}candidate filtering}

We define two modes:

\paragraph{Fixed\texorpdfstring{\nbhy}{-}candidate (liveness failures).}
Compute the static candidate set $S_k$ (size $C$) from the fixed ring; among $S_k$ choose the highest\nbhy scoring \emph{alive} node. This confines failover within a predetermined set.

\paragraph{All candidates down (rare) fallback.}
If all $C$ candidates are down, extend the window by another block of $C$ candidates (repeat) until finding an alive node or hitting an implementation cap (e.g., max\nbhy scan). This preserves availability; it may increase work only in extreme failure scenarios.

\section{Analysis}

\subsection{Complexity and memory}

\paragraph{Time.}
Lookup is $O(\log|\mathcal{R}|{+}C)$: one lower\nbhy bound search on $|\mathcal{R}|=NV$ plus $C$ next\nbhy distinct steps.

\paragraph{Space.}
Ring storage is $O(NV)$ entries. With $N{=}5000$, $V{=}256$, $|\mathcal{R}|{=}1.28$M.

\subsection{Churn guarantee under liveness failures}

\begin{theorem}[Zero Excess Churn Under Fixed\texorpdfstring{\nbhy}{-}Candidate Liveness Failures]
\label{thm:churn}
Assume the ring topology is unchanged and only node liveness changes. Under \LRH{} fixed\nbhy candidate filtering, only keys whose original winner is now dead are remapped; therefore excess churn is zero.
\end{theorem}

\begin{proof}[Proof sketch]
For any key $k$, the candidate set $S_k$ and the per\nbhy candidate scores are deterministic functions of $(k,\mathcal{R})$. If the original winner remains alive, it still has the maximal score among alive candidates in $S_k$, so the mapping does not change. Only keys whose original winner is dead fail over to the best alive candidate in the same $S_k$. Thus no healthy keys move, implying zero excess churn.
\end{proof}

\subsection{Smoothing Effect Analysis of \texorpdfstring{$C$}{C}-th Local Selection}
\label{sec:smoothing}

\paragraph{What we claim vs.\ what we do not claim.}
\LRH{} can be viewed as a \emph{local multi-choice} mechanism: each key compares $C$ \emph{distinct} ring-local candidates
and selects the maximum HRW score among them.
Classical ``power of two choices'' results assume (nearly) i.i.d.\ global samples; \LRH{} candidates are ring-local and hence correlated.
Therefore:
\begin{itemize}
  \item We \textbf{do not claim} a universal, assumption-free worst-case bound for all adversarial token layouts.
  \item We \textbf{do claim} a precise \emph{smoothing identity} under score symmetry, and a clean \emph{scaling law}
  under a standard randomized-ring regime.
\end{itemize}
In Appendix~D, we explicitly bound the secondary effects omitted by the simplified model, including the randomness of the candidate count $d_n$, discrete key sampling, and locality correlations.

\paragraph{Fluid model and notation.}
Let the ring contain $m=|\mathcal{R}|=NV$ tokens on $[0,1)$.
Let $G_1,\dots,G_m$ be consecutive cyclic gap lengths with $\sum_{i=1}^m G_i=1$.
All keys hashing into the same gap share the same successor index and therefore the same \LRH{} candidate set;
denote by $S_i\subseteq\mathcal{N}$ the $C$ distinct candidates associated with gap $i$.

\paragraph{Score symmetry.}
For a fixed key $k$, assume the per-node HRW scores $\{s(k,n)\}_{n\in S_i}$ are i.i.d.\ with a continuous distribution
(e.g., produced by a keyed hash). Ties occur with probability $0$.

\begin{lemma}[Uniform winner within a fixed candidate set]
\label{lem:uniform-winner-main}
Fix a key $k$ and a candidate set $S$ with $|S|=C$.
If $\{s(k,n)\}_{n\in S}$ are i.i.d.\ continuous, then for every $n\in S$,
\[
\Pr\!\bigl[\arg\max_{x\in S} s(k,x) = n\bigr] = \frac{1}{C}.
\]
\end{lemma}
\begin{proof}
By exchangeability under relabeling, each of the $C$ candidates is equally likely to be the unique maximizer.
\end{proof}

\paragraph{A smoothing identity (structural).}
Let $L_n$ be node $n$'s \emph{fluid} load share (fraction of the unit circle mapped to $n$).
Lemma~\ref{lem:uniform-winner-main} implies each gap contributes its mass evenly to all $C$ candidates in its set:
\begin{equation}
\label{eq:lrh-load-linear-main}
L_n \;=\; \frac{1}{C}\sum_{i:\, n\in S_i} G_i.
\end{equation}
Thus, relative to ring CH (where each gap maps entirely to one successor), \LRH{} replaces ``winner-takes-all''
by a \emph{local averaging operator} over $C$ candidates.

\paragraph{Scaling law.}
Under a randomized-ring regime (i.i.d.\ uniform token positions), \eqref{eq:lrh-load-linear-main} yields the familiar variance scale
$\mathrm{SD}(L_n)\approx 1/(N\sqrt{VC})$ and the max-deviation scale
$\mathrm{PALR}=1+O(\sqrt{\ln N/(VC)})$.
This confirms the $\sqrt{C}$ smoothing gain observed in our evaluations.

\paragraph{Note on baseline fairness: $V$ vs.\ $VC$ (cost matters).}
The randomized-ring scaling suggests that \LRH{}'s structural smoothing
$\mathrm{SD}(L_n)\propto 1/\sqrt{VC}$ matches ring CH if one simply increases vnodes to $V'=VC$.
While statistically equivalent, the distinction is cost structure: ring CH pays $C\times$ state blow-up (ring size $NV\to NVC$),
worsening cache/TLB locality and build overhead,
whereas \LRH{} keeps the ring size fixed at $NV$ and pays an extra $O(C)$ \emph{ring-local} work per lookup.
Our vnode sweep (Section~\ref{sec:ring-vn-sweep}) empirically illustrates this trade-off.

\subsection{Availability Analysis of \texorpdfstring{$p^C$}{p\textasciicircum C}}
\label{sec:availability-main}

\paragraph{What we claim vs.\ what we do not claim.}
Fixed-candidate liveness handling yields zero excess churn (Theorem~\ref{thm:churn}) but implies that availability depends on the joint failure probability of the candidate set.
We presents models for independent failures and fixed-fraction failures.
We \textbf{do not claim} intrinsic immunity to topology-correlated failures (e.g., rack failure) if tokens are placed naively; however, Appendix~D.7 shows that with standard random hashing (which decorrelates logical position from physical ID), the probability of hitting a single point of failure decays exponentially in $C$.

\begin{theorem}[Independent liveness model]
\label{thm:pc-main}
Assume each node is independently down with probability $p$.
For any key $k$ whose \LRH{} candidate set $S_k$ contains $C$ distinct nodes,
\emph{let $A_k$ be the event that all $C$ candidates in $S_k$ are down.} Then
\[
\Pr[A_k] = p^C,
\qquad
\Pr[\overline{A_k}] = 1-p^C.
\]
\end{theorem}
\begin{proof}
Distinctness gives $C$ independent down events; multiply probabilities.
\end{proof}

\begin{theorem}[Fixed-$F$ failure model (hypergeometric)]
\label{thm:hypergeo-main}
Suppose exactly $F$ out of $N$ nodes are down, and $S_k$ behaves as a uniformly random $C$-subset of nodes.
Then
\[
\Pr[S_k \subseteq \text{Failed}] = \frac{\binom{F}{C}}{\binom{N}{C}}
\;\le\;
\left(\frac{F}{N}\right)^C.
\]
\end{theorem}
\begin{proof}
Counting argument for the exact ratio; the inequality follows from
\[
\frac{\binom{F}{C}}{\binom{N}{C}}
=\prod_{j=0}^{C-1}\frac{F-j}{N-j}
\le (F/N)^C
\].
\end{proof}

\paragraph{Fallback cost.}
If the implementation extends the window by additional $C$-blocks until success,
under independent failures the number of blocks is geometric:
\[
\begin{aligned}
\mathbb{E}[\text{\#blocks}] &= \frac{1}{1-p^C},\\
\mathbb{E}[\text{\#candidates examined}] &= \frac{C}{1-p^C}.
\end{aligned}
\]

\section{Implementation Notes}

\paragraph{Hashing.}
Two hashes are used: \algo{HashPos} determines ring position and \algo{HashScore} determines \HRW{} score. In adversarial environments, these should be keyed (e.g., SipHash) \cite{aumasson2012siphash} to prevent crafted keys from skewing balance.

\paragraph{Modes in evaluation.}
We evaluate several update/failure semantics:

\begin{itemize}
  \item \textbf{rebuild}: rebuild the structure after failures/membership change.
  \item \textbf{next\nbhy alive} (ring/mpch variants): keep topology and scan to the next alive node during lookup.
  \item \textbf{fixed\nbhy candidate} (\LRH): compute a fixed candidate set and select the best alive within it.
\end{itemize}

These semantics matter: algorithms that rebuild can induce additional remapping (excess churn), while fixed\nbhy topology strategies can achieve zero excess churn for liveness changes (Theorem~\ref{thm:churn}).

\section{Evaluation}
\label{sec:eval}

\subsection{Experimental setup}

\paragraph{Hardware/OS.}
Windows~11, Intel Ultra~7 265KF, 64\,GB DDR5-6400.

\paragraph{Core parameters.}
$N{=}5000$ nodes, $V{=}256$ vnodes/node ($|\mathcal{R}|{=}1.28$M), $K{=}50{,}000{,}000$ keys, \LRH{} candidates $C{=}8$, Maglev table size $M{=}65537$ (prime), failure sizes $\{1,10,50\}$, 5 repeats per failure size, 1 warmup.

\paragraph{Parallelism.}
All experiments use Rayon with 20 worker threads (reported by the benchmark) \cite{rayon}.
We report aggregate throughput across threads.

\paragraph{Workload generation.}
Keys are generated via a seeded PRNG (base seed 20251226; repeats use different derived seeds). The large $K$ reduces sampling noise for distribution metrics.

\paragraph{Fairness and comparability.}
All schemes are evaluated under a shared harness: identical key generation, identical failure sets, and a unified metric implementation for moved/excess/affected/recv. We make evaluation semantics explicit in the results (\texttt{[rebuild]}, \texttt{[next-alive]}, \texttt{[fixed-cand]}), since failure-handling policy is part of the systems contract and directly affects churn. We also remove non-algorithmic overheads from timed regions where possible (e.g., allocating an \texttt{alive\_all} vector outside the timed query loop) to avoid ``implementation constant'' bias.

\paragraph{Implementation scope and threats to validity.}
Our implementations are reference-level and prioritize semantic fidelity over peak optimization (e.g., limited use of prefetching, vectorization, or layout-specific tuning). Consequently, throughput results should be interpreted as prototype-level evidence and may vary under optimized production implementations, although the observed gaps are consistent with the algorithms' differing memory-access patterns (sequential ring steps vs.\ scattered probes). Additionally, \HRW{} uses key sampling for large $N$ due to its $O(N)$ per-key cost, so its throughput is reported separately and is not directly comparable to full-key runs.

\subsection{Baselines}

We compare:
Ring \cite{karger1997consistent},
\MPCH \cite{appleton2015mpch},
Maglev \cite{eisenbud2016maglev},
Jump \cite{lamping2014jump},
\HRW \cite{thaler1998hrw},
and CRUSH \cite{weil2006crush}.
Finally, we implement a \emph{CRUSH-like} two-level rack model to serve as a structural baseline.

\paragraph{Note on Jump semantics.}
Jump Hash requires contiguous bucket IDs \cite{lamping2014jump}. Our benchmark models a rebuild\nbhy by\nbhy renumber semantics under membership changes, which is not Jump’s intended sweet spot; we report results for transparency and to illustrate semantic sensitivity.

\subsection{Metrics (exact definitions)}

Let $K_{\text{used}}$ be the number of keys actually evaluated in a row (e.g., sampled for full \HRW).

\begin{itemize}
  \item \textbf{Churn\%} $= \frac{\text{moved}}{K_{\text{used}}}\times 100\%$, where moved counts keys with $\texttt{init\_idx} \ne \texttt{fail\_idx}$.
  \item \textbf{Excess\%}: churn beyond the theoretical minimum for the given failure size (as computed by the benchmark harness).
  \item \textbf{FailAffected}: number of keys whose \emph{initial} assignment fell into the failed node set (keys that truly require failover), not equal to moved.
  \item \textbf{MaxRecvShare} $= \max_i \frac{\text{recv}[i]}{\text{affected}}$, where \texttt{recv[i]} counts only \emph{affected keys} remapped to alive node $i$.
  \item \textbf{Conc($\times$)} $= \frac{\text{MaxRecvShare}}{1/N_{\text{alive}}}$, i.e., the busiest receiver's share relative to ideal equal split among alive nodes.
  \item \textbf{ScanAvg}: average scan steps (each step checks the next token/candidate until hitting an alive node), computed as \texttt{scan\_sum/denom}.
  \item \textbf{ScanMax}: maximum observed scan steps in any single mapping (init or fail).
\end{itemize}

\subsection{Overall results (all failure sizes)}

Table~\ref{tab:overallA} reports the overall average across failure sizes (15 runs total).
Figure~\ref{fig:tradeoff-overall} visualizes the throughput/balance trade-off across the evaluated methods.

\begin{table*}[t]
\centering
\caption{Overall Average Across Failure Sizes (N=5000, V=256, K=50M; 15 runs total). HRW is reported on sampled keys; modes [rebuild]/[next-alive]/[fixed-cand] indicate different failure-handling semantics. Throughput is aggregate across 20 Rayon threads.}
\label{tab:overallA}
\scriptsize
\setlength{\tabcolsep}{3.6pt}
\begin{adjustbox}{max width=\textwidth}
\begin{tabular}{@{}l r r r r r r r r r r r r r@{}}
\toprule
\textbf{Algorithm} &
\textbf{K} &
\textbf{Build} &
\textbf{Query} &
\textbf{Thrpt} &
\textbf{Max/Avg} &
\textbf{P99/Avg} &
\textbf{CV} &
\textbf{Churn} &
\textbf{Excess} &
\textbf{FailAff} &
\textbf{MaxRecv} &
\textbf{Conc} &
\textbf{ScanAvg/Max} \\
& used & (ms) & (ms) & (M/s) & & & & (\%) & (\%) & & share & ($\times$) & \\
\midrule
Ring(vn=256)[rebuild] & 50M & 60.98 & 1303.94 & 38.36 & 1.2785 & 1.1550 & 0.0639 & 0.408 & 0.000 & 204044 & 0.0099 & 49.33 & 0.00/0 \\
Ring(vn=256)[next-alive] & 50M & 0.00 & 725.34 & 68.95 & 1.2785 & 1.1550 & 0.0639 & 0.408 & 0.000 & 204044 & 0.0099 & 49.33 & 1.00/3 \\
\MPCH{}(ring,vn=256,P=8)[next-alive] & 50M & 28.69 & 5680.67 & 8.80 & 1.0697 & 1.0439 & 0.0192 & 0.407 & 0.000 & 203587 & 0.0020 & 10.08 & 8.02/11 \\
\textbf{\LRH{}(vn=256,C=8)[fixed-cand]} & \textbf{50M} & \textbf{30.58} & \textbf{832.98} & \textbf{60.05} & \textbf{1.0947} & \textbf{1.0574} & \textbf{0.0244} & \textbf{0.407} & \textbf{0.000} & \textbf{203416} & \textbf{0.0012} & \textbf{6.14} & \textbf{8.00/8} \\
\LRH{}(vn=256,C=8)[rebuild] & 50M & 62.68 & 1509.80 & 33.16 & 1.0947 & 1.0574 & 0.0244 & 0.719 & 0.312 & 203416 & 0.0012 & 5.89 & 8.00/8 \\
Jump[rebuild-buckets] & 50M & 0.00 & 202.47 & 248.09 & 1.0361 & 1.0232 & 0.0100 & 66.214 & 65.808 & 203140 & 0.3740 & 1868.90 & 0.00/0 \\
	Maglev(M=65537)[rebuild] & 50M & 2.98 & 38.51 & 1301.76 & 1.1000 & 1.0818 & 0.0257 & 1.971 & 1.565 & 202824 & 0.0334 & 166.72 & 0.00/0 \\
	\HRW{}(sample K=2M) & 2M & 0.01 & 1723.36 & 1.16 & 1.1810 & 1.1185 & 0.0501 & 0.405 & 0.000 & 8106 & 0.0026 & 12.93 & 0.00/0 \\
	CRUSH-like(rack=50,bp=8,lp=8,tries=16) & 50M & 0.00 & 604.37 & 82.79 & 1.0379 & 1.0233 & 0.0100 & 0.406 & 0.000 & 203143 & 0.0005 & 2.58 & 16.03/57 \\
	\bottomrule
	\end{tabular}
	\end{adjustbox}
\end{table*}

\begin{figure*}[t]
\centering
\includegraphics[width=\textwidth]{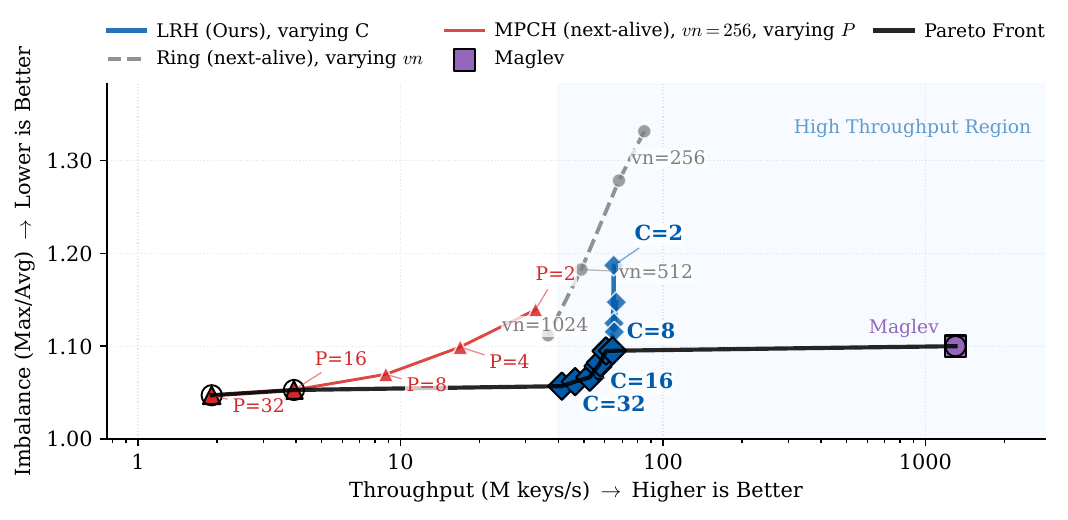}
\caption{Throughput--imbalance trade-off (all failure sizes): throughput vs.\ balance (PALR Max/Avg). We compare Ring (next-alive) varying virtual nodes per physical node $vn$ (gray dashed) against \LRH{} varying the local candidate count $C$ (blue). We additionally plot \MPCH{} on the same ring size ($vn{=}256$) while varying probe count $P\in\{2,4,8,16,32\}$ (red triangles), showing \MPCH{}'s own balance--throughput trade-off under identical failure-handling semantics.}
\label{fig:tradeoff-overall}
\end{figure*}

\paragraph{Uniformity.}
Ring(vn=256) has Max/Avg=1.2785. \LRH{} reduces this to 1.0947 using only $C=8$ distinct candidates, while \MPCH{} (P=8) reaches 1.0697. Thus \LRH{} approaches \MPCH{}'s balance while keeping candidate enumeration local and strictly bounded.

\paragraph{Throughput and cache locality.}
\LRH{} achieves 60.05 M keys/s, only about 13\% below Ring(next-alive) at 68.95 M keys/s, indicating that the additional $C$ candidate checks are largely cache-local. In contrast, \MPCH{} shows a steep probe-driven trade-off: at $vn{=}256$, increasing probes from $P{=}2$ to $P{=}32$ improves imbalance from 1.1393 to 1.0471 but reduces throughput from 32.74 to 1.91 M keys/s. At comparable imbalance to \LRH{} (e.g., \MPCH{} $P{=}4$ at 1.0991), \MPCH{} is still markedly slower (16.90 vs.\ 60.05 M keys/s), consistent with probe-driven scattered ring accesses.

\paragraph{\MPCH{} sensitivity to ring size.}
Appendix Figure~\ref{fig:mpch-vn-sweep} fixes $P{=}8$ and varies $vn\in\{1,16,256\}$: \MPCH{} operates without virtual nodes, but larger rings improve balance at a substantial throughput cost, consistent with ring footprint effects in addition to the dominant multi-probe lookup cost.

\paragraph{Churn, excess churn, and semantic comparability.}
Under fixed\nbhy candidate liveness failover, \LRH{} achieves 0\% excess churn (Theorem~\ref{thm:churn}) and empirically enforces ScanMax$=C$ (8), matching the algorithm’s ``exactly $C$ distinct steps'' semantics. Rebuild-based methods change candidate sets or lookup tables and therefore may incur excess churn (e.g., Maglev rebuild), consistent with Maglev’s stated disruption trade-offs \cite{eisenbud2016maglev}.

\FloatBarrier

\subsection{Microbenchmark: \MPCH{} probe generation vs.\ assignment}
\label{sec:mpch-micro}

A common concern is whether \MPCH{} appears slower due to an unfairly expensive probe-generation implementation. To isolate this, we add a report that measures (i) probe generation only and (ii) assign-only (full mapping) while switching from per-probe mix64 hashing to standard double-hashing.
For $N{=}5000$, $V{=}256$ ($|\mathcal{R}|{=}1.28$M), $P{=}8$, and 2M sampled keys, lower-bound search costs about $\lceil\log_2 |\mathcal{R}|\rceil \approx 21$ steps per probe, i.e., about $P \times 21 \approx 168$ random ring loads per key, or roughly $168 \times 16$B $\approx$ 2.62 KiB of ring-entry traffic per key.

\begin{table}[t]
\centering
\small
\caption{\MPCH{} probe-generation microbenchmark (2M keys; identical across cases). Double-hashing speeds up probe generation by $4.41\times$ but improves assign-only throughput by only $1.06\times$, indicating assignment is dominated by $P \times$ lower-bound ring traffic.}
\label{tab:mpch-probe-gen}
\setlength{\tabcolsep}{5pt}
\begin{tabular}{@{}l r@{}}
\toprule
\textbf{Case} & \textbf{Mkeys/s} \\
\midrule
Assign-only (mix64 probes) & 0.81 \\
Assign-only (double-hash probes) & 0.86 \\
Probe-gen only (mix64 probes) & 75.60 \\
Probe-gen only (double-hash probes) & 333.49 \\
\bottomrule
\end{tabular}
\end{table}

Figure~\ref{fig:mpch-micro} summarizes the microbenchmark breakdown for \MPCH{}.

\begin{figure}[tb]
\centering
\includegraphics[width=\columnwidth]{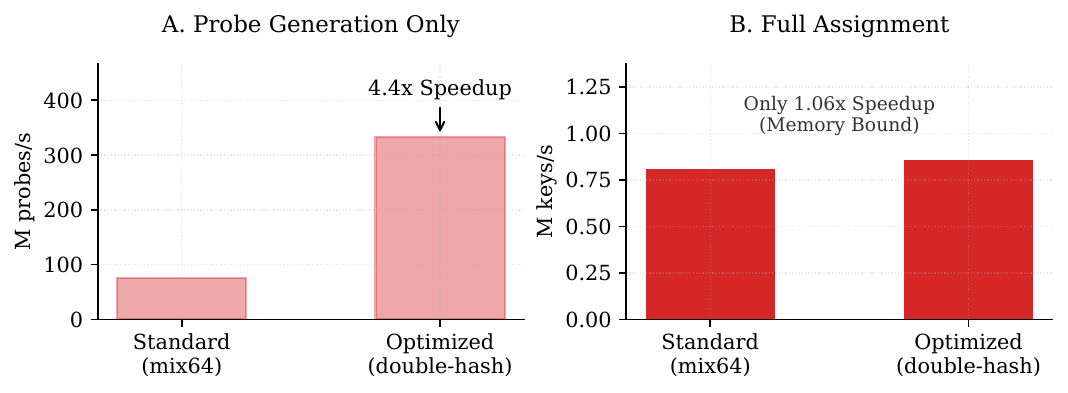}
\caption{\MPCH{} microbenchmark breakdown: probe generation can be sped up substantially, but assign-only throughput barely changes, indicating lower-bound ring traffic dominates.}
\label{fig:mpch-micro}
\end{figure}

This result strengthens the interpretation that \MPCH{} is fundamentally constrained by the microarchitectural cost of repeated lower-bound searches---scattered ring accesses with data-dependent branches and translation/cache pressure---not by hash arithmetic constants (Section~\ref{sec:vtune}).

\FloatBarrier

\subsection{Microarchitectural attribution with VTune}
\label{sec:vtune}

To connect the access-pattern difference to the throughput gap between \MPCH{} and \LRH{}, we collect Intel VTune hardware-event counters \cite{intelvtune,intelvtune2025guide} for the same large-scale topology ($N{=}5000$, $V{=}256$) and parameters ($P{=}8$, $C{=}8$), using $K{=}5$M keys to keep profiling overhead manageable.
Because the two methods have very different runtimes, we report \emph{normalized rates} (per branch, per cycle, or per retired instruction) rather than raw counts.

\begin{table}[tb]
\centering
\scriptsize
\caption{Selected VTune-derived rates from the same benchmark configuration (normalized to be comparable across different runtimes).}
\label{tab:vtune-norm}
\setlength{\tabcolsep}{3.5pt}
\begin{tabular}{@{}p{0.48\columnwidth} c r r r@{}}
\toprule
\textbf{Metric} & \textbf{Core} & \textbf{\MPCH{}} & \textbf{\LRH{}} & \textbf{Ratio} \\
\midrule
Branch mispredict rate ($\frac{\texttt{BR\_MISP}}{\texttt{BR\_INST}}$) & P & 0.161 & 0.072 & 2.22 \\
Branch mispredict rate ($\frac{\texttt{BR\_MISP}}{\texttt{BR\_INST}}$) & E & 0.158 & 0.091 & 1.74 \\
DTLB walk-active cycles / cycles ($\frac{\texttt{DTLB\_WALK\_ACTIVE}}{\texttt{CLK}}$) & P & 0.134 & 0.099 & 1.36 \\
L1D miss loads / inst retired ($\frac{\texttt{L1D\_MISS.LOAD}}{\texttt{INST}}$) & P & 0.075 & 0.046 & 1.63 \\
Offcore data-read cycles / cycles ($\frac{\texttt{OFFCORE\_DATA\_RD}}{\texttt{CLK}}$) & P & 0.656 & 0.598 & 1.10 \\
TopDown Bad Speculation share (E-core, $\frac{\texttt{BAD\_SPEC}}{\texttt{sum}}$) & E & 0.583 & 0.350 & 1.67 \\
\bottomrule
\end{tabular}
\end{table}

\begin{figure*}[t]
\centering
\includegraphics[width=\textwidth]{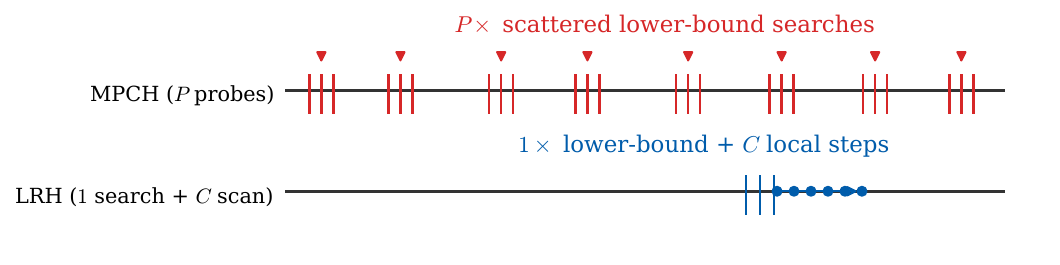}
\caption{Schematic ring-access traces: \MPCH{} performs $P$ scattered lower-bound searches; \LRH{} performs one lower-bound search plus $C$ local ring steps.}
\label{fig:access-schematic}
\end{figure*}

\begin{figure}[tb]
\centering
\includegraphics[width=\columnwidth]{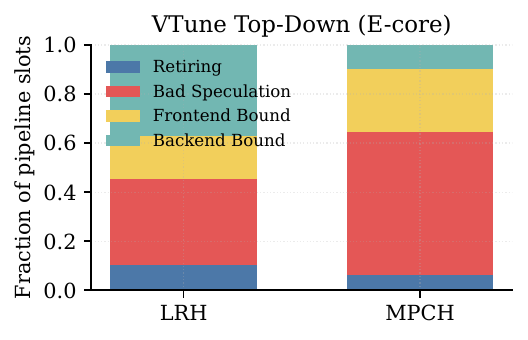}
\caption{VTune Top-Down breakdown on E-cores (normalized by pipeline slots).}
\label{fig:vtune-topdown}
\end{figure}

\paragraph{Unpredictable control flow from repeated binary searches.}
\MPCH{} performs $P$ independent lower-bound searches per key; each search is a data-dependent binary search over a large ring array and therefore induces unpredictable branches.
Consistent with this structure, Table~\ref{tab:vtune-norm} shows substantially higher branch misprediction rates for \MPCH{} than \LRH{} (2.22$\times$ on P-cores; 1.74$\times$ on E-cores).

\paragraph{Translation and cache pressure from scattered probes.}
Binary search touchpoints are scattered across the ring, so repeated searches amplify TLB walks and cache-miss exposure.
In our VTune runs, \MPCH{} exhibits higher DTLB walk-active cycles (1.36$\times$) and higher L1D miss-load density (1.63$\times$ per retired instruction) on P-cores, with a modest increase in offcore data-read cycles (1.10$\times$).
These indicators align with the microbenchmark result in Section~\ref{sec:mpch-micro}: even large improvements to probe generation do not change the dominant cost in the assignment phase.

\paragraph{Why Top-Down shares can be counterintuitive.}
In VTune's Top-Down taxonomy, work executed on wrong-path speculation may still generate cache/TLB pressure, yet the resulting stall time is attributed to \emph{Bad Speculation} rather than \emph{Memory Bound}.
Accordingly, Figure~\ref{fig:vtune-topdown} shows that \MPCH{} shifts a larger fraction of pipeline slots into Bad Speculation on E-cores, which can make a simple ``memory-bound share'' argument incomplete for multi-probe binary-search workloads.

\paragraph{Measurement scope.}
The VTune logs underlying Table~\ref{tab:vtune-norm} and Figure~\ref{fig:vtune-topdown} are collected at the process level for the benchmark configuration and include repeated runs; we therefore emphasize normalized rates that are insensitive to wall-time differences.
For per-lookup attribution within only the timed query loop, VTune supports scoping counters via ITT task markers; we leave this refinement to future artifact updates.

\subsection{Ablation: candidate count $C$}
\label{sec:ablation-c}

We vary the number of distinct candidates $C$ under all-alive conditions. Larger $C$ improves balance but reduces throughput, matching the intended trade-off.

\begin{table}[tb]
\centering
\small
\caption{\LRH{} ablation over $C$ (all-alive). Increasing $C$ improves Max/Avg but decreases throughput.}
\label{tab:ablation-c}
\setlength{\tabcolsep}{5pt}
\begin{tabular}{@{}r r r@{}}
\toprule
\textbf{$C$} & \textbf{Max/Avg} & \textbf{Thrpt (M/s)} \\
\midrule
2  & 1.1871 & 64.95 \\
4  & 1.1248 & 65.09 \\
8  & 1.0947 & 60.68 \\
16 & 1.0679 & 52.51 \\
32 & 1.0569 & 41.20 \\
\bottomrule
\end{tabular}
\end{table}

\begin{figure}[tb]
\centering
\includegraphics[width=\columnwidth]{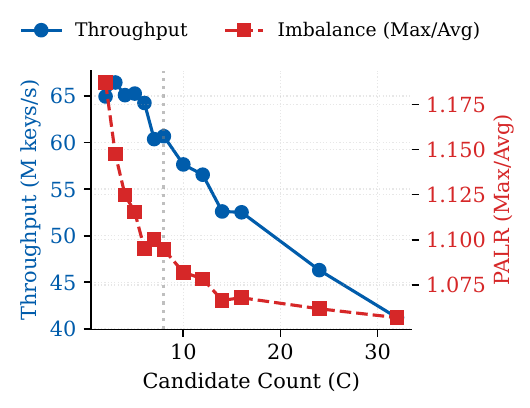}
\caption{\LRH{} ablation: increasing candidate count $C$ improves balance (Max/Avg) at the cost of throughput.}
\label{fig:ablation-c}
\end{figure}

\subsection{Ablation: ring vnode count $V$}
\label{sec:ring-vn-sweep}

To quantify how far classical ring hashing can be pushed by increasing vnodes alone, we sweep the number of virtual nodes per physical node $V$ for Ring under the same large-scale setup ($N{=}5000$, $K{=}50$M, averaged across failure sizes) and the same failure-handling policy (\texttt{[next-alive]}).
Figure~\ref{fig:ring-vn-sweep} shows that increasing $V$ substantially improves balance, but exhibits clear diminishing returns: Max/Avg decreases from 2.6914 ($V{=}8$) to 1.1118 ($V{=}1024$), while beyond $V{\ge}128$ each doubling improves Max/Avg by only about 4--8\% (e.g., 1.3316$\to$1.2785$\to$1.1826$\to$1.1118).
At the same time, throughput drops sharply: 131.62 M keys/s ($V{=}8$) $\to$ 67.98 M keys/s ($V{=}256$) $\to$ 36.50 M keys/s ($V{=}1024$).
We report the sweep\nbhy run throughput here; the overall suite run in Table~\ref{tab:overallA} reports 68.95 M keys/s at $V{=}256$ (within 1.4\%).
This supports a cache/locality interpretation: the ring array grows as $|\mathcal{R}|{=}N\!\cdot\!V$ (a $4\times$ increase from $V{=}256$ to $V{=}1024$), making binary searches less cache-friendly and reducing effective memory locality even under \texttt{next-alive}.

Build cost also scales with ring size as expected for ring construction: Ring(\texttt{[rebuild]}) build time grows from 1.55\,ms ($V{=}8$) to 303.09\,ms ($V{=}1024$, ${\approx}195\times$).
Meanwhile, churn remains essentially unchanged across $V$ (about 0.40\% in our runs), consistent with the fact that the fraction of keys that must move is primarily determined by the failure fraction rather than token granularity; vnodes mainly influence \emph{how evenly} affected keys are redistributed.

Finally, this sweep strengthens the comparison against simply ``turning the vnode knob'': achieving Ring Max/Avg $\approx 1.11$ requires $V{=}1024$ and yields 36.50 M keys/s, whereas \LRH{} at $V{=}256$ achieves better balance (Max/Avg 1.0947) at higher throughput (60.05 M keys/s, ${\approx}1.65\times$).

\begin{figure}[tb]
\centering
\includegraphics[width=\columnwidth]{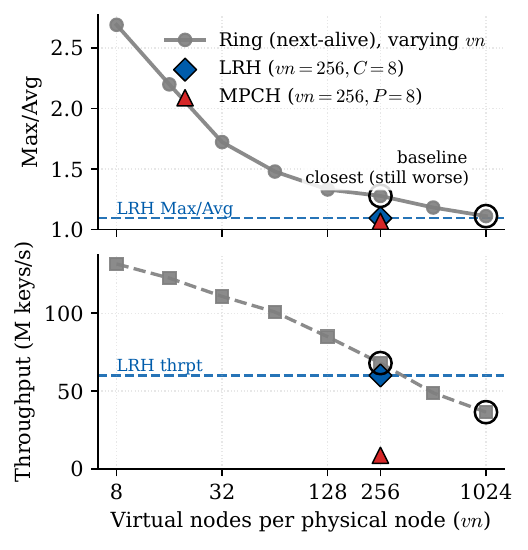}
\caption{Ring vnode sweep (next-alive): increasing vnodes improves balance with diminishing returns but reduces throughput; overlay points show \LRH{} and \MPCH{} at $V{=}256$ for reference.}
\label{fig:ring-vn-sweep}
\end{figure}

\subsection{Churn by failure size}

\begin{table}[tb]
\centering
\small
\caption{Churn and Excess Churn by Failure Size (K=50M).}
\label{tab:churn}
\setlength{\tabcolsep}{5pt}
\begin{tabular}{@{}l c c c@{}}
\toprule
\textbf{Algorithm} & \textbf{F=1} & \textbf{F=10} & \textbf{F=50} \\
\midrule
\multicolumn{4}{@{}l}{\textit{Churn\%}} \\
Ring [next-alive]           & 0.020 & 0.201 & 1.004 \\
\textbf{\LRH{} [fixed-cand]} & \textbf{0.020} & \textbf{0.200} & \textbf{1.000} \\
\addlinespace
\multicolumn{4}{@{}l}{\textit{Excess\%}} \\
Ring [next-alive]            & 0.000 & 0.000 & 0.000 \\
\textbf{\LRH{} [fixed-cand]} & \textbf{0.000} & \textbf{0.000} & \textbf{0.000} \\
\LRH{} [rebuild]             & 0.015 & 0.155 & 0.765 \\
Maglev [rebuild]             & 0.145 & 1.037 & 3.513 \\
Jump [rebuild-renum]         & 16.136 & 83.632 & 97.657 \\
\bottomrule
\end{tabular}
\end{table}

\subsection{Discussion: what the numbers say}

\paragraph{Why fixed\nbhy candidate works.}
Fixed\nbhy candidate filtering isolates liveness failures from topology changes, enabling Theorem~\ref{thm:churn}. Rebuild changes the candidate sets (or lookup tables), creating additional remapping (excess churn).

\paragraph{Concentration of failover load.}
\LRH{} reduces concentration compared to ring next\nbhy alive: overall Conc($\times$) is 6.14 vs 49.33. This means the busiest receiver of affected keys is closer to an ideal equal split among alive nodes.

Figure~\ref{fig:conc-by-f} shows how concentration evolves with failure size and highlights \LRH{}'s reduced receiver hot-spotting.

\begin{figure}[tb]
\centering
\includegraphics[width=\columnwidth]{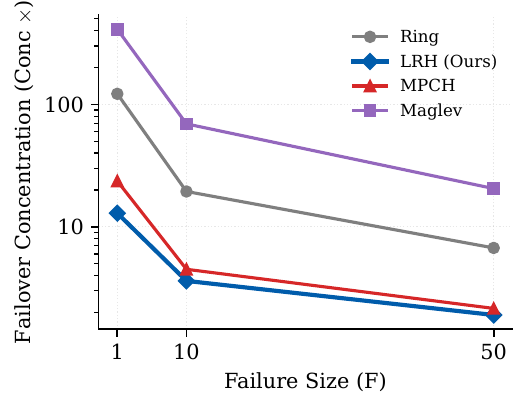}
\caption{Failover concentration (Conc$\times$) vs.\ failure size $F$ (log scale). \LRH{} consistently reduces receiver hot-spotting compared to ring next-alive.}
\label{fig:conc-by-f}
\end{figure}

\FloatBarrier

\paragraph{Interpretation of Jump.}
Under rebuild\nbhy by\nbhy renumber semantics, Jump shows extreme churn. This reflects a semantic mismatch: Jump is designed for sequential bucket IDs and is not intended for arbitrary deletion/renumber membership semantics \cite{lamping2014jump}.

\subsection{Membership changes: add/remove 1\% of nodes}
\label{sec:membership}

Our strongest guarantee targets liveness changes under fixed topology. For completeness, we also measure membership changes that rebuild data structures. We add/remove 1\% of nodes ($N{=}5000 \rightarrow 5050$ and $N{=}5000 \rightarrow 4950$).

\begin{table}[tb]
\centering
\small
\caption{Membership change (rebuild semantics): churn and excess churn under $\pm$1\% node changes.}
\label{tab:membership-1pct}
\setlength{\tabcolsep}{5pt}
\begin{tabular}{@{}l r r@{}}
\toprule
\textbf{Algorithm} & \textbf{Churn\%} & \textbf{Excess\%} \\
\midrule
\multicolumn{3}{@{}l}{\textit{+1.00\% nodes (5000 $\rightarrow$ 5050)}} \\
\LRH{}(vn=256,C=8) & 1.750 & 0.760 \\
Ring(vn=256) & 0.992 & 0.000 \\
Maglev(M=65537) & 4.247 & 3.331 \\
\addlinespace
\multicolumn{3}{@{}l}{\textit{-1.00\% nodes (5000 $\rightarrow$ 4950)}} \\
\LRH{}(vn=256,C=8) & 1.766 & 0.765 \\
Ring(vn=256) & 1.004 & 0.000 \\
Maglev(M=65537) & 4.511 & 3.513 \\
\bottomrule
\end{tabular}
\end{table}

Figure~\ref{fig:membership-1pct} visualizes churn/excess churn under $\pm$1\% membership changes.

\begin{figure}[tb]
\centering
\includegraphics[width=\columnwidth]{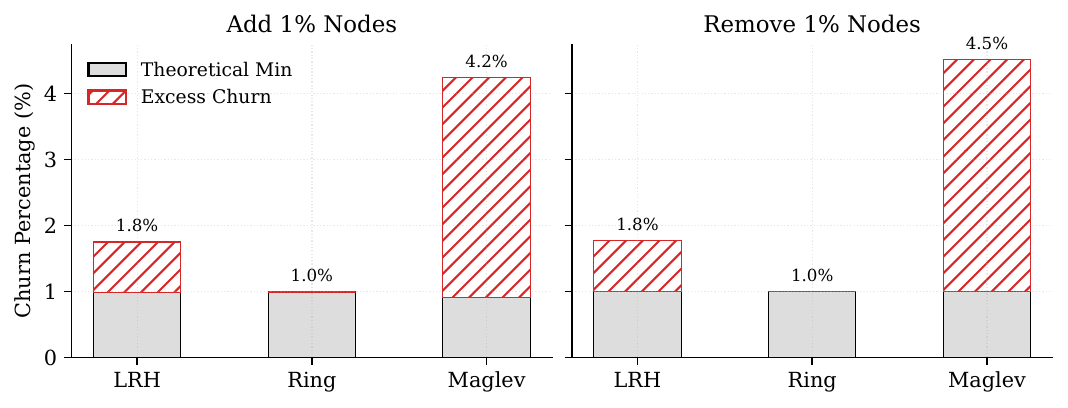}
\caption{Membership change (rebuild semantics): churn and excess churn under $\pm$1\% node changes.}
\label{fig:membership-1pct}
\end{figure}

These results highlight a semantic distinction: fixed-topology liveness handling can achieve zero excess churn (Theorem~\ref{thm:churn}), whereas membership changes that rebuild candidate sets or tables may introduce excess churn depending on the algorithm.

\section{Limitations and Future Work}

\paragraph{Binary search locality.}
Although candidate enumeration is strictly $C$ ring steps, the lower\nbhy bound search costs $O(\log|\mathcal{R}|)$ accesses. Future work can explore cache\nbhy friendly layouts (e.g., Eytzinger layout) \cite{khuong2017arraylayouts,khuong2015arraylayouts} or coarse indexing to reduce this cost.

\paragraph{Membership changes.}
Our strongest churn guarantee is for liveness failures under fixed topology. Membership changes alter the ring and may induce additional churn. We report $\pm$1\% membership changes in Section~\ref{sec:membership}; broader sweeps (larger percentages and more workloads) remain important for a full comparison with \AnchorHash{} \cite{mendelson2021anchorhash} and table\nbhy based approaches \cite{eisenbud2016maglev}.
More broadly, \AnchorHash{} is explicitly optimized for \emph{arbitrary membership} updates; our membership experiments use a ring rebuild semantics and should not be read as a like\nbhy for\nbhy like evaluation of \AnchorHash{}'s strengths in that regime.

\paragraph{Weighted accuracy.}
Weighted \HRW{} is standard \cite{ietfWeightedHRW}. However, in \LRH{} weights interact with the candidate generation mechanism. A dedicated weighted evaluation (e.g., Zipf weights, bimodal capacities) is planned to quantify allocation error vs $C$.

\paragraph{Implementation optimization.}
Our implementations prioritize semantic fidelity over peak optimization. Throughput results may vary under optimized production implementations; future work includes more aggressive optimizations and comparisons with highly tuned baselines.

\section{Related Work}

\paragraph{Consistent hashing and vnodes.}
Consistent hashing originates with Karger et al. \cite{karger1997consistent}. \MPCH{} provides a widely cited analysis of imbalance and vnode cost \cite{appleton2015mpch}. Jump discusses constraints (sequential buckets) and contrasts memory behavior with ring/vnode tables \cite{lamping2014jump}.

\paragraph{Rendezvous hashing and weights.}
\HRW{} (Rendezvous hashing) is described by Thaler and Ravishankar \cite{thaler1998hrw}. Weighted variants are standardized in an IETF draft \cite{ietfWeightedHRW}.

\paragraph{Bounded loads.}
\BLCH{} bounds maximum load under updates using forwarding and provides rigorous expected move bounds with a tunable parameter $c=1+\epsilon$ \cite{mirrokni2018bounded}.

\paragraph{Modern consistent hashing systems.}
\AnchorHash{} provides a scalable consistent hash with strong consistency under arbitrary membership changes \cite{mendelson2021anchorhash}. \LRH{} targets a different trade\nbhy off: bounded candidate work and cache\nbhy local lookup paths with explicit liveness failure semantics. In ring\nbhy centric deployments, this can translate into a smaller implementation delta: \LRH{} reuses the existing ring layout and changes only the \emph{choice rule} (a $C$\nobreakdash-way local \HRW{} election), whereas adopting \AnchorHash{} may require replacing ring\nbhy specific data structures and operational conventions.

\paragraph{Table\nbhy based balancers.}
Maglev \cite{eisenbud2016maglev} achieves near\nbhy perfect balance via a lookup table and characterizes disruption under backend changes.

\paragraph{Multiple choice intuition.}
Classical ``power of two choices'' results explain why comparing a small set of candidates reduces imbalance under suitable sampling assumptions \cite{mitzenmacher2001twochoices}. \LRH{} differs in that candidates are ring\nbhy local and engineered for cache locality.

\paragraph{Topology-aware placement.}
CRUSH \cite{weil2006crush} is a decentralized hierarchy\nbhy aware placement algorithm; our CRUSH\nbhy like baseline models a rack hierarchy.

\section{Conclusion}

\LRH{} reconciles a common systems tension: ring\nbhy based schemes are stable but can be imbalanced without large vnode state, while multi\nbhy probe/table approaches improve balance at the cost of scattered probes or disruption under rebuild. \LRH{} restricts \HRW{} to a cache\nbhy local window of $C$ distinct neighbors and uses next\nbhy distinct offsets to bound enumeration to exactly $C$ ring steps. In a large\nbhy scale benchmark ($N{=}5000$, $V{=}256$, $K{=}50$M, $C{=}8$), \LRH{} reduces Max/Avg from 1.2785 to 1.0947 while preserving 0\% excess churn under fixed\nbhy candidate liveness failover, and empirically enforces ScanMax$=C$ via next\nbhy distinct offsets. \LRH{} achieves 60.05 M keys/s, only modestly below Ring(next-alive) at 68.95 M keys/s, and provides ${\approx}6.8\times$ the throughput of \MPCH{} (P=8) in our implementation while approaching \MPCH{}'s load balance.

\bibliographystyle{plain}
\bibliography{refs}

\appendix

\section*{Appendix A: Reproduction Parameters}

\begin{lstlisting}
--nodes 5000 --keys 50000000 --vnodes 256
--candidates 8 --maglev-m 65537
--fail-list 1,10,50 --repeats 5 --warmup 1
--seed 20251226 --threads 0
--hrw-full-max-n 2000 --hrw-sample-keys 2000000
--max-scan 4096 --mp-probes 8
--crush-rack-size 50 --crush-bucket-probes 8
--crush-leaf-probes 8 --crush-tries 16
--report-memory --report-mpch-probe-gen --report-membership
--membership-pct 1.0 --ablation-c-list 2,4,8
\end{lstlisting}

\section*{Appendix B: \MPCH{} Sensitivity to Ring Size}

\vspace{-0.6em}
\begin{figure}[H]
\centering
\includegraphics[width=\columnwidth]{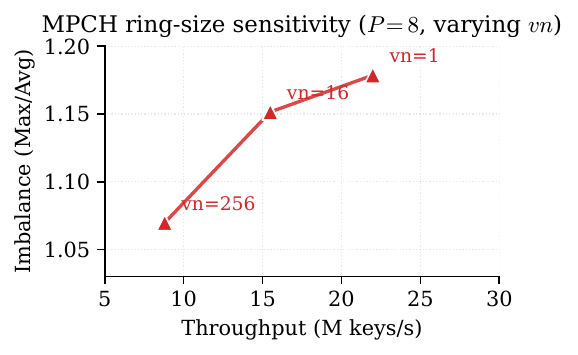}
\caption{\MPCH{} sensitivity to ring size (vn). We fix $P{=}8$ and vary $vn\in\{1,16,256\}$, isolating ring-size effects from the multi-probe mechanism.}
\label{fig:mpch-vn-sweep}
\end{figure}

\FloatBarrier
\onecolumn

\section*{Appendix C: Per-Failure Detailed Tables}

\begin{center}
\scriptsize
\setlength{\tabcolsep}{3.2pt}
\captionof{table}{Detailed Results (failed\_nodes=1), averaged over 5 repeats (K=50M).}
\label{tab:fail1}
\begin{adjustbox}{max width=\textwidth}
\begin{tabular}{@{}l r r r r r r r r r r r r r r@{}}
\toprule
Algorithm & K & Build & Query & Thrpt & Max/Avg & P99/Avg & CV & Churn\% & Excess\% & FailAff & MaxRecv & Conc & ScanAvg & ScanMax \\
\midrule
Ring(vn=256)[rebuild] & 50M & 60.57 & 1304.71 & 38.33 & 1.2785 & 1.1550 & 0.0639 & 0.020 & 0.000 & 9948 & 0.0244 & 121.86 & 0.00 & 0 \\
Ring(vn=256)[next-alive] & 50M & 0.00 & 717.00 & 69.74 & 1.2785 & 1.1550 & 0.0639 & 0.020 & 0.000 & 9948 & 0.0244 & 121.86 & 1.00 & 2 \\
MPCH(ring,vn=256,P=8)[next-alive] & 50M & 30.20 & 5660.70 & 8.84 & 1.0697 & 1.0439 & 0.0192 & 0.020 & 0.000 & 10075 & 0.0047 & 23.61 & 8.00 & 10 \\
\textbf{LRH(vn=256,C=8)[fixed-cand]} & \textbf{50M} & \textbf{30.01} & \textbf{819.54} & \textbf{61.02} & \textbf{1.0947} & \textbf{1.0574} & \textbf{0.0244} & \textbf{0.020} & \textbf{0.000} & \textbf{9927} & \textbf{0.0026} & \textbf{12.90} & \textbf{8.00} & \textbf{8} \\
LRH(vn=256,C=8)[rebuild] & 50M & 64.53 & 1511.57 & 33.12 & 1.0947 & 1.0574 & 0.0244 & 0.035 & 0.015 & 9927 & 0.0025 & 12.40 & 8.00 & 8 \\
Jump[rebuild-buckets] & 50M & 0.00 & 210.29 & 238.02 & 1.0361 & 1.0232 & 0.0100 & 16.156 & 16.136 & 10035 & 1.0000 & 4999.00 & 0.00 & 0 \\
	Maglev(M=65537)[rebuild] & 50M & 3.12 & 37.39 & 1339.74 & 1.1000 & 1.0818 & 0.0257 & 0.165 & 0.145 & 9825 & 0.0821 & 410.42 & 0.00 & 0 \\
	HRW(sample K=2M) & 2M & 0.01 & 1742.41 & 1.15 & 1.1810 & 1.1185 & 0.0501 & 0.020 & 0.000 & 395 & 0.0056 & 28.02 & 0.00 & 0 \\
		CRUSH-like(rack=50,bp=8,lp=8,tries=16) & 50M & 0.00 & 601.80 & 83.17 & 1.0379 & 1.0233 & 0.0100 & 0.020 & 0.000 & 9994 & 0.0009 & 4.40 & 16.00 & 45 \\
			\bottomrule
			\end{tabular}
			\end{adjustbox}
\end{center}

\par\vspace{10pt}
\begin{center}
\scriptsize
\setlength{\tabcolsep}{3.2pt}
\captionof{table}{Detailed Results (failed\_nodes=10), averaged over 5 repeats (K=50M).}
\label{tab:fail10}
\begin{adjustbox}{max width=\textwidth}
\begin{tabular}{@{}l r r r r r r r r r r r r r r@{}}
\toprule
Algorithm & K & Build & Query & Thrpt & Max/Avg & P99/Avg & CV & Churn\% & Excess\% & FailAff & MaxRecv & Conc & ScanAvg & ScanMax \\
\midrule
Ring(vn=256)[rebuild] & 50M & 60.03 & 1310.51 & 38.18 & 1.2785 & 1.1550 & 0.0639 & 0.201 & 0.000 & 100424 & 0.0039 & 19.41 & 0.00 & 0 \\
Ring(vn=256)[next-alive] & 50M & 0.00 & 723.17 & 69.14 & 1.2785 & 1.1550 & 0.0639 & 0.201 & 0.000 & 100424 & 0.0039 & 19.41 & 1.00 & 3 \\
MPCH(ring,vn=256,P=8)[next-alive] & 50M & 28.34 & 5713.26 & 8.75 & 1.0697 & 1.0439 & 0.0192 & 0.201 & 0.000 & 100316 & 0.0009 & 4.49 & 8.01 & 11 \\
\textbf{LRH(vn=256,C=8)[fixed-cand]} & \textbf{50M} & \textbf{31.66} & \textbf{833.66} & \textbf{59.98} & \textbf{1.0947} & \textbf{1.0574} & \textbf{0.0244} & \textbf{0.200} & \textbf{0.000} & \textbf{100123} & \textbf{0.0007} & \textbf{3.61} & \textbf{8.00} & \textbf{8} \\
LRH(vn=256,C=8)[rebuild] & 50M & 58.39 & 1509.31 & 33.16 & 1.0947 & 1.0574 & 0.0244 & 0.355 & 0.155 & 100123 & 0.0007 & 3.36 & 8.00 & 8 \\
Jump[rebuild-buckets] & 50M & 0.00 & 205.89 & 244.34 & 1.0361 & 1.0232 & 0.0100 & 83.831 & 83.632 & 99780 & 0.1015 & 506.39 & 0.00 & 0 \\
	Maglev(M=65537)[rebuild] & 50M & 2.88 & 38.37 & 1304.23 & 1.1000 & 1.0818 & 0.0257 & 1.237 & 1.037 & 99890 & 0.0139 & 69.23 & 0.00 & 0 \\
	HRW(sample K=2M) & 2M & 0.01 & 1715.48 & 1.17 & 1.1810 & 1.1185 & 0.0501 & 0.199 & 0.000 & 3989 & 0.0015 & 7.51 & 0.00 & 0 \\
		CRUSH-like(rack=50,bp=8,lp=8,tries=16) & 50M & 0.00 & 605.60 & 82.61 & 1.0379 & 1.0233 & 0.0100 & 0.200 & 0.000 & 99884 & 0.0004 & 1.92 & 16.02 & 58 \\
			\bottomrule
			\end{tabular}
			\end{adjustbox}
\end{center}

\par\vspace{10pt}
\begin{center}
\scriptsize
\setlength{\tabcolsep}{3.2pt}
\captionof{table}{Detailed Results (failed\_nodes=50), averaged over 5 repeats (K=50M).}
\label{tab:fail50}
\begin{adjustbox}{max width=\textwidth}
\begin{tabular}{@{}l r r r r r r r r r r r r r r@{}}
\toprule
Algorithm & K & Build & Query & Thrpt & Max/Avg & P99/Avg & CV & Churn\% & Excess\% & FailAff & MaxRecv & Conc & ScanAvg & ScanMax \\
\midrule
Ring(vn=256)[rebuild] & 50M & 62.34 & 1296.59 & 38.57 & 1.2785 & 1.1550 & 0.0639 & 1.004 & 0.000 & 501760 & 0.0014 & 6.71 & 0.00 & 0 \\
Ring(vn=256)[next-alive] & 50M & 0.00 & 735.85 & 67.95 & 1.2785 & 1.1550 & 0.0639 & 1.004 & 0.000 & 501760 & 0.0014 & 6.71 & 1.01 & 4 \\
MPCH(ring,vn=256,P=8)[next-alive] & 50M & 27.55 & 5668.07 & 8.82 & 1.0697 & 1.0439 & 0.0192 & 1.001 & 0.000 & 500370 & 0.0004 & 2.15 & 8.04 & 13 \\
\textbf{LRH(vn=256,C=8)[fixed-cand]} & \textbf{50M} & \textbf{30.06} & \textbf{845.75} & \textbf{59.16} & \textbf{1.0947} & \textbf{1.0574} & \textbf{0.0244} & \textbf{1.000} & \textbf{0.000} & \textbf{500198} & \textbf{0.0004} & \textbf{1.90} & \textbf{8.00} & \textbf{8} \\
LRH(vn=256,C=8)[rebuild] & 50M & 65.11 & 1508.54 & 33.20 & 1.0947 & 1.0574 & 0.0244 & 1.766 & 0.765 & 500198 & 0.0004 & 1.93 & 8.00 & 8 \\
Jump[rebuild-buckets] & 50M & 0.00 & 191.22 & 261.92 & 1.0361 & 1.0232 & 0.0100 & 98.656 & 97.657 & 499604 & 0.0205 & 101.30 & 0.00 & 0 \\
	Maglev(M=65537)[rebuild] & 50M & 2.93 & 39.79 & 1261.32 & 1.1000 & 1.0818 & 0.0257 & 4.511 & 3.513 & 498759 & 0.0041 & 20.50 & 0.00 & 0 \\
	HRW(sample K=2M) & 2M & 0.01 & 1712.20 & 1.17 & 1.1810 & 1.1185 & 0.0501 & 0.997 & 0.000 & 19935 & 0.0007 & 3.28 & 0.00 & 0 \\
		CRUSH-like(rack=50,bp=8,lp=8,tries=16) & 50M & 0.00 & 605.71 & 82.60 & 1.0379 & 1.0233 & 0.0100 & 0.999 & 0.000 & 499552 & 0.0003 & 1.41 & 16.08 & 67 \\
			\bottomrule
			\end{tabular}
			\end{adjustbox}
\end{center}
\twocolumn

\section*{Appendix D: Critical Modeling Boundaries and More Rigorous Derivations}
\label{app:rigor}

This appendix refines the main-text scaling arguments by making explicit:
(i) randomness of the candidate-coverage count $d_n$ (compound variance),
(ii) discrete key sampling noise vs.\ structural imbalance,
(iii) locality-induced correlations, and
(iv) the impact of rack-correlated failures.
The goal is to \emph{bound} the terms omitted by simplified derivations and clarify when they are negligible.

\subsection*{D.1\quad Setup and the structural smoothing identity}

Let $m=NV$ be the number of ring tokens and let $G_1,\dots,G_m$ be the cyclic gaps.
For each gap $i$, let $S_i$ be the $C$ distinct \LRH{} candidates.
Let $L_n$ denote node $n$'s \emph{fluid} load share.

\begin{lemma}[Uniform HRW winner within a set]
\label{lem:uniform-winner-app}
Fix a key $k$ and a candidate set $S$ with $|S|=C$.
If $\{s(k,n)\}_{n\in S}$ are i.i.d.\ continuous random variables, then each $n\in S$ wins with probability $1/C$.
\end{lemma}

\begin{lemma}[Local averaging (load linearization)]
\label{lem:load-linear-app}
Under Lemma~\ref{lem:uniform-winner-app}, the fluid load satisfies
\[
L_n \;=\; \frac{1}{C}\sum_{i:\, n\in S_i} G_i.
\]
\end{lemma}

\subsection*{D.2\quad Dirichlet gaps and the missing ``compound variance'' term}

Assume token positions are i.i.d.\ uniform on $[0,1)$ and then sorted. Then
\begin{equation}
\label{eq:dirichlet}
(G_1,\dots,G_m)\sim\mathrm{Dirichlet}(1,\dots,1).
\end{equation}

Let $I_n:=\{i: n\in S_i\}$ and $d_n:=|I_n|$.
A subtlety is that $d_n$ is generally \emph{random} (depending on local token layout and the ``distinct'' constraint),
so $L_n$ is a mixture over $d_n$.
We now make the variance decomposition explicit.

\begin{proposition}[Conditional Beta and total variance decomposition]
\label{prop:compound-var}
Under \eqref{eq:dirichlet}, conditional on $I_n$ (equivalently on $d_n$),
\[
\begin{aligned}
\sum_{i\in I_n} G_i \,\big|\, I_n &\sim \mathrm{Beta}(d_n,\, m-d_n),\\
L_n \,\big|\, d_n &\sim \frac{1}{C}\,\mathrm{Beta}(d_n,\, m-d_n).
\end{aligned}
\]
Moreover,
\begin{equation}
\label{eq:total-variance}
\mathrm{Var}(L_n)
=
\mathbb{E}\!\left[\mathrm{Var}(L_n\mid d_n)\right]
+
\mathrm{Var}\!\left(\mathbb{E}[L_n\mid d_n]\right).
\end{equation}
\end{proposition}

\paragraph{Why the second term exists.}
Because $\mathbb{E}[\,\mathrm{Beta}(d,m-d)\,]=d/m$, we have
\[
\mathbb{E}[L_n\mid d_n]=\frac{d_n}{C m}.
\]
Therefore the ``missing'' compound term is
\begin{equation}
\label{eq:missing-term}
\mathrm{Var}\!\left(\mathbb{E}[L_n\mid d_n]\right)
=
\frac{\mathrm{Var}(d_n)}{C^2 m^2}.
\end{equation}

\subsection*{D.3\quad Bounding $\mathrm{Var}(d_n)$ when $C\ll N$}

The distinct-candidate rule implies $d_n$ can drop below $VC$ if multiple tokens of the same node
fall within the same ``$C$-distinct'' neighborhood, causing overlapping coverage intervals.
We do not attempt an exact law for $d_n$; instead we give a conservative bound.

\paragraph{A simple collision proxy.}
For each token occurrence of node $n$, consider the event that node $n$ appears again
before the scan collects $C$ \emph{distinct} nodes ahead.
Let $Y_t\in\{0,1\}$ be the indicator for the $t$-th token of node $n$ that such a ``self-collision'' occurs.
Define $Y:=\sum_{t=1}^V Y_t$.
A self-collision implies that the coverage contribution of the colliding token partially overlaps with its predecessor.
Conservatively, each such event reduces the ideal coverage count $VC$ by at most $C$. Thus,
\begin{equation}
\label{eq:deficit-bound}
0 \le VC-d_n \le C\,Y.
\end{equation}

\begin{lemma}[Self-collision probability scale]
\label{lem:collision-scale}
Under uniform random token placement and $C\ll N$, one has
\[
\mathbb{E}[Y_t] \;\lesssim\; \frac{C}{N},
\qquad\text{hence}\qquad
\mathbb{E}[Y] \;\lesssim\; \frac{VC}{N}.
\]
\end{lemma}

\paragraph{Variance bound (rare-collision regime).}
Lemma~\ref{lem:collision-scale} gives $\mathbb{E}[Y]\lesssim VC/N$.
Moreover, in the same randomized-ring regime collisions are rare and local; consequently $Y$ behaves like a
Poisson-binomial count and concentrates at its mean. A conservative second-moment bound yields
\[
\mathrm{Var}(Y)=O(\mathbb{E}[Y])=O\!\left(\frac{VC}{N}\right).
\]
This holds, for example, when $VC/N=O(1)$.
Plugging into \eqref{eq:deficit-bound} gives:
\begin{equation}
\label{eq:var-dn-bound}
\mathrm{Var}(d_n)
\le
C^2 \mathrm{Var}(Y)
\;\lesssim\;
\frac{V C^3}{N}.
\end{equation}

\paragraph{Impact on $\mathrm{Var}(L_n)$.}
Plugging \eqref{eq:var-dn-bound} into \eqref{eq:missing-term} with $m=NV$ yields
\[
\mathrm{Var}\!\left(\mathbb{E}[L_n\mid d_n]\right)
=\frac{\mathrm{Var}(d_n)}{C^2 m^2}
\;\lesssim\;
\frac{VC^3/N}{C^2 (N^2V^2)}
=\frac{C}{N^3V}.
\]
Compared to the leading structural term
$\mathbb{E}[\mathrm{Var}(L_n\mid d_n)]\approx \mathbb{E}[d_n]/(C^2 m^2)\approx 1/(N^2VC)$,
the ratio is on the order of $C^2/N$.
Equivalently, since both terms share the denominator $C^2m^2$, the compound/structural ratio is
$\mathrm{Var}(d_n)/\mathbb{E}[d_n]=O(C^2/N)$ under \eqref{eq:var-dn-bound}.
Thus, for large clusters ($C^2 \ll N$), the randomness of $d_n$ is a negligible second-order effect.

\subsection*{D.4\quad Fluid load vs.\ discrete keys}

The main text analyzes \emph{structural} imbalance ($L_n$).
Real systems have $K$ discrete keys, introducing sampling noise ($X_n$).
\begin{equation}
\label{eq:key-var-decomp}
\begin{aligned}
\mathrm{Var}(X_n)
&=
K\,\mathbb{E}[L_n(1-L_n)] \quad\text{(sampling)}\\
&\quad+\;
K^2\,\mathrm{Var}(L_n) \quad\text{(structural)}.
\end{aligned}
\end{equation}
For large $K$ per node ($K/N \gg 1$), the structural term dominates, justifying our focus on smoothing $L_n$.

\subsection*{D.5\quad Locality correlations}

Neighboring nodes have overlapping candidate windows, making $(L_n)$ positively correlated.
This does \emph{not} invalidate union bounds used for PALR (Theorem 5), as $\Pr[\cup A_n] \le \sum \Pr[A_n]$ holds regardless of independence.
Correlation implies that if overload occurs, it may manifest as a local cluster of loaded nodes rather than isolated spikes.

\subsection*{D.6\quad Baseline fairness: statistical equivalence vs.\ cost}

We emphasize that while \LRH{} is statistically equivalent to increasing vnodes to $V'=VC$,
it achieves this without the $C\times$ memory and build-time penalty associated with expanding the ring structure.
It decouples the smoothing parameter ($C$) from the ring state size ($NV$).

\subsection*{D.7\quad Rack-correlated failures}
\label{sec:rack-correlation-app}

Theorem~\ref{thm:pc-main} assumes independent failures.
In practice, failures may be correlated by topology (e.g., a top-of-rack switch failure).
Let the cluster consist of $N$ nodes partitioned into racks of size $R$.
Suppose a full rack fails (batch failure of $R$ nodes).

Because tokens are placed via a hash function that is agnostic to rack ID (Random Hashing),
the relative positions of tokens from a specific rack are random uniform.
Consequently, the set of candidates $S_k$ for any key $k$ (which consists of $C$ logically consecutive distinct nodes)
is effectively a random sample of size $C$ from the population of $N$ nodes, with respect to rack affiliation.

We can thus apply Theorem~\ref{thm:hypergeo-main} (Fixed-$F$ model) with $F=R$:
\[
\Pr[\text{all } C \text{ candidates are in the failed rack}]
\;\approx\;
\left(\frac{R}{N}\right)^C.
\]
For typical values (e.g., $R=40, N=5000, C=8$), this probability is $(0.008)^8 \approx 10^{-17}$, which is negligible.
This result confirms that \LRH{} maintains high availability under rack failures without requiring explicit cross-rack placement constraints, provided standard random hashing is used.

\end{document}